\documentclass{amsart}
\usepackage{tikz}
\usepackage{caption}
\usetikzlibrary{arrows}
\usepackage{amssymb, amscd, amsfonts}
\usepackage[mathscr]{eucal}
\usepackage{mathrsfs}
\usepackage{graphicx}
\usepackage{subfigure}
\graphicspath{{img/}}
\usepackage{lscape}
\usepackage{booktabs}
\newtheorem{theorem}{Theorem}[section]

\newtheorem{proposition}[theorem]{Proposition}

\theoremstyle{definition}

\newtheorem{example}{Example}

\theoremstyle{remark}
\newtheorem*{remark}{Remark}

\newcommand{\remove}[1]{}

\newcommand{\ie}{\textit{i.e.}}

\begin{document}

\title[A simple combinatorial interpretation of certain generalized Bell and Stirling]{A simple combinatorial interpretation of certain generalized Bell and Stirling numbers}

%%% Authors

\author{Pietro Codara, Ottavio M. D'Antona, Pavol Hell}

%%% Date

\date{\today}

%%% Addresses

\address[P. Codara, O. M. D'Antona]
{Dipartimento di Informatica, Universit\`{a} degli Studi di Milano
via Comelico 39/41, I-20135 Milano, Italy}
\email{\{codara,dantona\}@di.unimi.it}
\address[P. Hell]
{School of Computing Science, Simon Fraser University, Burnaby, B.C., Canada, V5A 1S6}
\email{pavol@sfu.ca}

\keywords{}

%\subjclass[2000]{Primary: . Secondary: .}
%%% Abstract
\begin{abstract}
In a series of papers, P. Blasiak et al. developed a wide-ranging
generalization of Bell numbers (and of Stirling numbers
of the second kind) that appears to be relevant to the so-called
\emph{Boson normal ordering problem}.
They provided a recurrence and, more recently, also offered a
(fairly complex) combinatorial interpretation
of these numbers. We show that by restricting the numbers somewhat
(but still widely generalizing Bell and Stirling
numbers), one can supply a much more natural combinatorial
interpretation. In fact, we offer two different such
interpretations, one in terms of graph colourings and another one
in terms of certain labelled Eulerian digraphs.
\end{abstract}
\maketitle

\section{Introduction}
\label{sec:intro}

In \cite{blasiak1,blasiak2,blasiak3,blasiak4} P. Blasiak et al.
introduced coefficients $B_{r,s}(n)$, and $S_{r,s}(n,k)$ that
provide a wide-ranging generalization of Bell numbers, and of
Stirling numbers of the second kind, respectively. In particular
they defined the generalized Bell polynomial (see
\cite[Equations (1.5) and (2.1)]{blasiak1})\footnote{We denote
by $(x)_n$ the falling factorial $x(x-1)\cdots (x-n+1)$.
(Note that the authors of \cite{blasiak1,blasiak2,blasiak3,blasiak4}
use the symbol $x^{\underline{n}}$ instead.)}

\begin{equation}\label{eq:gen_Bell_poly}
\begin{split}
B_{r,s}(n,t)& =\sum_{k=s}^{ns} S_{r,s}(n,k)t^k\ =\\
&=e^{-t}\sum_{k=0}^{\infty}\frac{1}{k!}\prod_{j=1}^n((k+(j-1)(r-s))_st^k,
\end{split}
\end{equation}

\noindent where $r, s, n, k$ are positive integers and $r \geq s$.

These coefficients generalize Bell numbers, and Stirling numbers of
the second kind, usually denoted $B_n$, and $S(n,k)$, respectively,
because by letting $r=s=t=1$ in the above formula, one obtains the
classical formula of Dobinski \cite{comtet}

\begin{equation}\label{eq:dobinski}
B_{1,1}(n)=\frac{1}{e}\sum_{k=0}^{\infty}\frac{k^n}{k!}.
\end{equation}

\noindent In fact $B_n=B_{1,1}(n)$, and $S(n,k)=S_{1,1}(n,k).$

\medskip
The work of P. Blasiak et al. was motivated by the fact that their
coefficients appear to be relevant for the so called \emph{Boson
normal ordering problem}.

\medskip
In \cite{blasiak2} the authors asked for a combinatorial
interpretation of these coefficients. Later on, in \cite{blasiak4},
they provided one such interpretation, in terms of what they called
\emph{colonies of bugs}. We refer to \cite[Section III]{blasiak4} for
the exact definition, but we remark that a colony of bugs is a fairly
complex object that corresponds to a labelled tree whose vertices
include {\em labels} as well as {\em cells}. Each bug in a colony
corresponds to a subtree, and has a {\em  type} $(r,s)$; it consists
of a body of $r$ cells, as well as of $s$ {\em legs}, some of which
can be {\em free}  \cite[Section III]{blasiak4}. It turns out
(\cite[Theorem 3.1]{blasiak4}) that $B_{r,s}(n)$ counts the number
of colonies of $n$ bugs each of type $(r,s)$, and that $S_{r,s}(n,k)$
counts the number of such colonies having exactly $k$ free legs.

In this note we suggest a simpler combinatorial interpretation of
these coefficients, at least in some important cases. Our interpretations
are stated in standard combinatorial terminology, in terms of colourings
and labeled Eulerian digraphs.

\medskip
Our focus is the case $r=s$. We supply two simple combinatorial
interpretations of the coefficients $B_{m,m}(n)$ and $S_{m,m}(n,k)$, for
all positive integers $m, n, k$. We note that these coefficients are still
much more general than the Bell numbers $B_{1,1}(n)$ and the Stirling
numbers of the second kind $S_{1,1}(n,k)$. Our first interpretation
(Section \ref{sec:stablepart}) is in terms of colourings of a certain graph.
In Sections \ref{sec:cycles} we supply another interpretation of the same
numbers in terms of the number of certain labeled Eulerian digraphs.
Finally, in Section \ref{sec:21case} we remark that in the general case
when $r$ and $s$ are different, there appear to be in certain cases
well-known simple combinatorial interpretations as well; we discuss
mostly the case $r=2, s=1$, but also remark on possible connections
for certain values in the cases $r>2$ and $s=1$.

\section{Colourings}
\label{sec:stablepart}

A $k$-{\em colouring} of a graph $G$ is a partition of the vertex set of $G$
into $k$ non-empty stable sets, \ie{} sets not containing adjacent vertices.
Each such stable set is called a {\em colour-class} of the partition.

Sometimes a $k$-colouring is defined as a mapping of vertices into a set of
$k$ colours, so that adjacent vertices obtain different colours. We note that
for us the names of the colours do not play a role, \ie, two mappings that
yield the same partition are considered the same colouring. Moreover, we
require that each colour-class is non-empty (which corresponds to the
requirement that each colour is used).

We denote by $K_m$ the complete graph on $m$ vertices, and by
$nK_m$ the disjoint union of $n$ copies of $K_m$.

For positive integers $ m,n, k$, let $C_{m}(n,k)$ denote the number of
$k$-colourings of $nK_m$. We first prove a recurrence for the numbers
$C_{m}(n,k)$.

\begin{proposition}
\label{prop:rec1}
We have
\begin{equation}
\label{eq:formula-rec1}
C_{m}(n,k)=\sum_{i=0}^m\binom{m}{i}(k-i)_{m-i} C_{m}(n-1,k-i)\,,
\end{equation}
with initial conditions
\begin{align*}
C_{m}(n,k)& = 0 \text{ whenever } k<m\,,\ \ and\\
C_{m}(1,k)& = \begin{cases}
1 & \text{if }\ k=m\,, \\
0 & \text{otherwise}\,.
\end{cases}
\end{align*}
\end{proposition}
\begin{proof}
The case $k<m$ is trivial (with fewer then $m$ colours we cannot colour $K_m$).
It is also obvious that when $n=1$ we have a unique $k$-colouring of $K_m$ when
$k=m$, and none when $k>m$.

To prove the recurrence, we describe how to obtain, in two steps, all $k$-colourings
of $nK_m$, for $k\geq m$ and $n\neq 1$. Fix an arbitrary copy of $K_m$.

\smallskip
(1) Choose $i$ vertices of the fixed $K_m$, each forming a singleton colour-class.

\smallskip
(2) Insert the remaining $m-i$ vertices of the fixed $K_m$ in the colour-classes of all
$(k-i)$-colourings of $(n-1)K_m$.

\smallskip
Step (1) can be done in $\binom{m}{i}$ ways, and step (2) in
$(k-i)_{m-i} C_{m}(n-1,k-i)$ ways. Our claim is proved.
\end{proof}

\begin{remark}
Note that  $C_{m}(n,nm)=1$.
\end{remark}

We now have the following result.

\begin{proposition}
\label{prop:colouring_Smm}
$S_{m,m}(n,k)$ counts the number of $k$-colourings of $nK_m$.
In other words, $S_{m,m}(n,k)=C_{m}(n,k)$.
 \end{proposition}
\begin{proof}
A simple manipulation of the formulas shows that recurrence
(\ref{eq:formula-rec1}) coincides with the recurrence (21) in %cite as (20-21) with initial condition
\cite{blasiak2}, namely:
\begin{align*}
% &S_{r,r}(1,r)=1,\ \ \ S_{r,r}(n,k)=0\ \ \text{ for } k<r\,,\ nr<k\leq (n+1)r\,,\\
%
&S_{r,r}(n+1,k)=\sum_{p=0}^r\binom{k+p-r}{p}(r)_p S_{r,r}(n,k+p-r)\,.%,\ r\leq k\leq nr\,,\ n>1\,.
\end{align*}
Indeed,
$(m)_i\binom{k+i-m}{i}=(k+i-m)_i\binom{m}{i}$.
\end{proof}

Needless to say, the recurrence (\ref{eq:formula-rec1}) generalizes
the classical recursion for the Stirling numbers of the second
kind. Using (\ref{eq:formula-rec1}), we can compute a few examples.
In Table \ref{tab:S_3} we compute the number of $k$-colourings of $nK_3$.
%
%\begin{table}[h]\footnotesize
% \begin{center}
%  \begin{tabular}{c|rrrrrrrrrrr}
%  \toprule
%  $S_{2,2}(n,k)$& \text{k=2} & 3 & 4 & 5 & 6 & 7 & 8 & 9 & 10 & 11 \\
%  \midrule
%   \text{n=1} & 1 & \text{} & \text{} & \text{} & \text{} & \text{} & \text{} & \text{} & \text{} & \text{} \\
% 2 & 2 & 4 & 1 & \text{} & \text{} & \text{} & \text{} & \text{} & \text{} & \text{} \\
% 3 & 4 & 32 & 38 & 12 & 1 & \text{} & \text{} & \text{} & \text{} & \text{} \\
% 4 & 8 & 208 & 652 & 576 & 188 & 24 & 1 & \text{} & \text{} & \text{} \\
% 5 & 16 & 1280 & 9080 & 16944 & 12052 & 3840 & 580 & 40 & 1 & \text{} \\
% 6 & 32 & 7744 & 116656 & 412800 & 540080 & 322848 & 98292 & 16000 & 1390 & 60 \\
%  \bottomrule
%  \end{tabular}
% \end{center}
% \caption{$S_{2,2}(n,k)$.}
% \label{tab:S_2}
%\end{table}
\begin{table}[h]\footnotesize
 \begin{center}
  \begin{tabular}{c|rrrrrrrr}
  \toprule
  $S_{3,3}(n,k)$ & \text{k=3} & 4 & 5 & 6 & 7 & 8 & 9 & 10 \\
  \midrule
 \text{n=1} & 1 & \text{} & \text{} & \text{} & \text{} & \text{} & \text{} & \text{} \\
 2 & 6 & 18 & 9 & 1 & \text{} & \text{} & \text{} & \text{} \\
 3 & 36 & 540 & 1242 & 882 & 243 & 27 & 1 & \text{} \\
 4 & 216 & 13608 & 94284 & 186876 & 149580 & 56808 & 11025 & 1107 \\
 5 & 1296 & 330480 & 6148872 & 28245672 & 49658508 & 41392620 & 18428400 & 4691412 \\
  \bottomrule
  \end{tabular}
 \end{center}
 \caption{$S_{3,3}(n,k)$.}
 \label{tab:S_3}
\end{table}

%Tables \ref{tab:S_2} and \ref{tab:S_3} also appear in \cite[Table 1]{blasiak1}.
%(Table \ref{tab:S_2} is A078739 in \cite{oeis}.)
Table \ref{tab:S_3} also appears in \cite[Table 1]{blasiak1}.
Denoting by $B_m(n)$ the number of \emph{all} colourings of $n K_m$, we have (cf. \cite[Equation (1.5)]{blasiak1}):
\begin{equation}\label{eq:Bm}
B_m(n) = \sum_{k=m}^{nm} C_{m}(n,k) = \sum_{k=m}^{nm} S_{m,m}(n,k) = B_{m,m}(n)\,.
\end{equation}
%For instance, summing the rows of Table \ref{tab:S_2} we obtain
%$1, 7, 87, 1657, 43833, \dots$, that is the sequence $B_{2,2}(n)$ in \cite{blasiak1}.
For instance, summing the rows of Table \ref{tab:S_3} we obtain
$1, 34, 2971, 513559, \dots$, that is the sequence $B_{3,3}(n)$ in \cite{blasiak1},
that is the sequence A069223 from \cite{oeis}.
\begin{example}Figure \ref{fig:3k3} shows the graph $2K_3$.

\bigskip
\begin{center}
\begin{tikzpicture}[auto,node distance=2cm,
  thick,main node/.style={circle,fill=blue!20,draw,font=\sffamily\Large\bfseries}]
\node[main node] (1) {a};
  \node[main node] (2) [above right of=1] {b};
  \node[main node] (3) [below right of=2] {c};
  \node[main node] (4) [ right of=3] {d};
\node[main node] (5) [ above right of=4] {e};
  \node[main node] (6) [ below right of=5] {f};
    \path[color= red, every node/.style={color=red}]
    (1)  edge [] node[above] {} (2)
   (2)  edge [] node[below] {} (3)
 (1)  edge [] node[above] {} (3)
  (4)  edge [] node[above] {} (5)
   (5)  edge [] node[below] {} (6)
 (4)  edge [] node[above] {} (6);
  \end{tikzpicture}
   \captionof{figure}{The graph $2 K_3$.}
   \label{fig:3k3}
  \end{center}

\smallskip
%The six $3$-colourings of $2K_3$ are:
%\begin{align*}
%ad|be|cf & & ad|bf|ce & & ae|bd|cf & & ae|bf|cd & & af|bd|ce & & af|be|cd\,.
%\end{align*}

The eighteen $4$-colourings of $2K_3$ are
\begin{align*}
a|d|be|cf & & a|d|bf|ce & & a|e|bd|cf & & a|e|bf|cd & & a|f|bd|ce & & a|f|be|cd\, \\
ad|b|e|cf & & ad|b|f|ce & & ae|b|d|cf & & ae|b|f|cd & & af|b|d|ce & & af|b|e|cd\, \\
ad|be|c|f & & ad|bf|c|e & & ae|bd|c|f & & ae|bf|c|d & & af|bd|c|e & & af|be|c|d\,.
\end{align*}
%The $9$ stable partitions of $2K_3$ in $5$ blocks, which correspond to its $5$-colourings, are:
%\begin{align*}
%ad|b|c|e|f & ; & ae|b|c|d|f & ; & af|b|c|d|e & ; & a|bd|c|e|f & ; & a|be|c|d|f & ; & a|bf|c|d|e\,; \\
%a|b|cd|e|f & ; & a|b|ce|d|f & ; & a|b|cf|d|e & . & & & & &
%\end{align*}
%Finally, the unique stable partitions of $2K_3$ in $6$ blocks, which corresponds to its $6$-colouring,
%is $a|b|c|d|e|f$. In total we have $34$ stable partitions of $2K_3$, which correspond to all
%its possible colourings.
\end{example}

\section{Labelled Eulerian Digraphs}
\label{sec:cycles}

We consider digraphs that allow loops and multiple edges in the same direction.
A digraph $G$ is \emph{Eulerian} if at every vertex the in-degree equals
the out-degree. (Note that we do not require $G$ to be connected.)
The edge set of an Eulerian digraph $G$ can be partitioned into directed cycles.
We call an Eulerian digraph \emph{$(n,m)$-labelled} if
its edge set is partitioned into $n$ directed $m$-cycles, each with a distinguished
first edge (and hence a unique second, third, etc.,
$m$-th edge). Figure 2 shows a $(2,3)$-labelled Eulerian digraph, with its 2
directed 3-cycles; the $j^{th}$ edge of the $i^{th}$ cycle is labelled $e_{i,j}$.

\medskip
\begin{center}
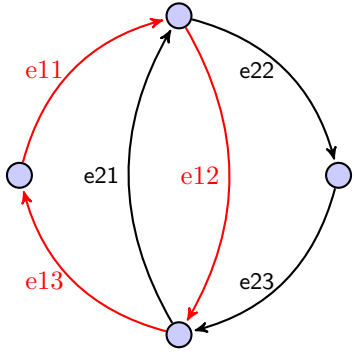

\begin{tikzpicture}[->,>=stealth',shorten >=1pt,auto,node distance=3cm,
  thick,main node/.style={circle,fill=blue!20,draw,font=\sffamily\Large\bfseries}]

  \node[main node] (2) {};
  \node[main node] (1) [below left of=2] {};
  \node[main node] (3) [below right of=1] {};
  \node[main node] (4) [below right of=2] {};

  \path[color= red, every node/.style={color=red}]
    (1)  edge [bend left] node[left] {e11} (2)
    (2)  edge [bend left]node[left] {e12} (3)
    (3)  edge [bend left]node[left] {e13} (1);
 \path[every node/.style={font=\sffamily\small}]
    (2)  edge [bend left]node[left] {e22} (4)
    (3)  edge [bend left]node[left] {e21} (2)
    (4)  edge [bend left]node[left] {e23} (3);
\end{tikzpicture}
\captionof{figure}{A $(2,3)$-labelled Eulerian digraph.}
\label{fig:T_3(2)}
\end{center}

\begin{theorem}
\label{Pavol}
The number of $(n,m)$-labelled Eulerian digraphs is equal to $B_{m,m}(n)$.
\end{theorem}

\begin{proof}
We show a bijection between the set of $(n,m)$-labelled Eulerian digraphs
and the number of colourings of $n K_m$. To this end we assign an
arbitrary order to the $n$ cliques of $nK_m$. Thus the vertices of $nK_m$ will be
called $v_{i,j}$ for  $i=1, 2, \dots, n$, and $j=1, 2, \cdots, m$. We define a
bijective mapping $\phi$ associating $e_{i,j}$ with $v_{i,j}$. (Here $e_{i,j}$ is the
$i^{th}$ edge of the $j^{th}$ cycle.)

\medskip\noindent$\bullet$
From graphs to colourings. Let $\mathcal{T}_{m}(n)$ be the set of
$(n,m)$-labelled Eulerian digraphs.
Here we establish a bijection between the $k$-colourings
of $nK_m$ and the elements of  $\mathcal{T}_{m}(n)$ with $k$ vertices.  Let now $\tau$ be an element of $\mathcal{T}_{m}(n)$
with $k$ vertices. Let , for $t=1, 2, \dots, k$, $B_t$ be the set of
edges of $\tau$ that are incident
in vertex $t$. It is obvious that
\[
\{B_1, B_2, \dots, B_k\}
\]
is a partition of the set of edges of  $\tau$. Now, by construction, one sees that
\[
\{\phi(B_1), \phi(B_2), \dots, \phi(B_k)\}
\]
is a $k$-colouring of $nK_m$.

For instance, the graph drawn in the picture \ref{fig:T_3(2)} corresponds to the following
colouring of $2K_3$
\[
v_{1,3}\ |\ v_{1,1} v_{2,1}\ |\ v_{1,2} v_{2,3\ }|\ v_{2,2}
\]

\medskip\noindent$\bullet$
From colourings to graphs. Let $\pi=\{B_1, B_2, ..., B_k\}$ be a colouring of
$nK_m$. We describe the directed graph, $\tau$, associated with $\pi$.

\smallskip
$\tau$ has $k$ vertices, say $w_1, w_2, ...,w_k$. To define the edges of $\tau$ we
assume first $m>1$. Let, for $i=1, 2, ..., n$, and $j=1, 2, ..., m$, $B_{p}$ be the
block of $\pi$ containing vertex $v_{i,j}$, and $B_{q}$ be the block of $\pi$ containing
vertex $v_{i,j+1}$. Notice that the indices of the vertices of $nK_m$ are considered
in clockwise order: $v_{i,m+1}\equiv v_{i,1}$. Then edge $e_{i,j}$ starts at $w_{p}$,
and ends at $w_{q}$.

If $m=1$, the edges of $\tau$ are loops. Specifically $e_{i,1}$ starts and ends at
$w_{t}$, where $t$ is the index of the block of $\pi$ containing vertex $v_{i,1}$.
\end{proof}

Thus we can say again that the number of $(n,m)$-labelled Eulerian digraphs with
$k$ vertices enjoy the same recurrence as $S_{m}(n,k)$.
Therefore counting these graphs corresponds to another combinatorial interpretations
of the coefficients of \cite{blasiak1}.

\medskip
We close the Section with a remark. It is obvious that any $k$-colouring of a given
set is fully described by any $k-1$ of its colour-classes. Accordingly, one can give
a slightly different interpretation of coefficients $S_{m}(n,k)$ by removing the last
edge from each cycle, producing a partition into labeled directed paths instead of
cycles. This model generalizes the concept of loopless, oriented multigraphs on $n$
labeled arcs as in A020556 in \cite{oeis}.

\section{Conclusions}
\label{sec:21case}

We hope that simpler combinatorial interpretations can be found for other generalized
Bell numbers and Stirling numbers of the second kind. In particular, we note that
our bijections (in Sections 2 and 3) exist for the disjoint union of cliques of different
sizes.

For the coefficients $S_{2,1}(n,k)$ we observe that Equation (15) of \cite{blasiak2}
implies that $S_{2,1}(n,k)$ is equal to the (positive) Lah number
\[
L(n,k)=\frac{n!}{k!}\binom{n-1}{k-1}\,.
\]
According to the classical interpretation of Lah numbers, this means that $S_{2,1}(n,k)$
counts the number of ordered placements of $n$ balls into $k$ boxes, and $B_{2,1}(n)$
counts the number of ordered placements of $n$ balls into boxes \cite{comtet}.

Table \ref{tab:S_21} provides some values of $S_{2,1}(n,k)$. Those values also appear in
\cite[Table 1]{blasiak1}, and in sequences A105278 of \cite{oeis}, where
further combinatorial interpretations of such coefficients are proposed.

\begin{table}[h]\footnotesize
 \begin{center}
  \begin{tabular}{c|rrrrrrrrr}
  \toprule
  $S_{2,1}(n,k)$ & \text{k=1} & 2 & 3 & 4 & 5 & 6 & 7 & 8 & 9 \\
  \midrule
 \text{n=1} & 1 & \text{} & \text{} & \text{} & \text{} & \text{} & \text{} & \text{} & \text{} \\
 2 & 2 & 1 & \text{} & \text{} & \text{} & \text{} & \text{} & \text{} & \text{} \\
 3 & 6 & 6 & 1 & \text{} & \text{} & \text{} & \text{} & \text{} & \text{} \\
 4 & 24 & 36 & 12 & 1 & \text{} & \text{} & \text{} & \text{} & \text{} \\
 5 & 120 & 240 & 120 & 20 & 1 & \text{} & \text{} & \text{} & \text{} \\
 6 & 720 & 1800 & 1200 & 300 & 30 & 1 & \text{} & \text{} & \text{} \\
 7 & 5040 & 15120 & 12600 & 4200 & 630 & 42 & 1 & \text{} & \text{} \\
 8 & 40320 & 141120 & 141120 & 58800 & 11760 & 1176 & 56 & 1 & \text{} \\
 9 & 362880 & 1451520 & 1693440 & 846720 & 211680 & 28224 & 2016 & 72 & 1 \\
  \bottomrule
  \end{tabular}
 \end{center}
 \caption{$S_{2,1}(n,k)$.}
 \label{tab:S_21}
\end{table}

%Finally, we remark that the values of $S_{3,1}(n,1)$ in Table 1 in \cite{blasiak1} appear related to the sequence
%A001147 from \cite{oeis} , which counts the number of increasing ordered rooted trees on $n+1$ vertices. (Here
%"increasing" means the vertices are labeled $0, 1, 2,..., n$ so that each path from the root has increasing labels.)
%The sequence A007559 should be similarly related to the values $S_{4,1}(n,1)$.

Finally, we remark that the values of $S_{3,1}(n,1)$ in Table 1 in \cite{blasiak1} appear to be identical to the sequence
A001147 from \cite{oeis}, which counts the number of increasing ordered rooted trees on $n+1$ vertices. (Here
"increasing" means the vertices are labeled $0, 1, 2,..., n$ so that each path from the root has increasing labels.)
Similarly, the values $S_{4,1}(n,1)$  appear to be identical to the sequence A007559 from \cite{oeis}.

%\bibliographystyle{amsalpha}
%\bibliography{bb}

\newcommand{\etalchar}[1]{$^{#1}$}
\providecommand{\bysame}{\leavevmode\hbox to3em{\hrulefill}\thinspace}
\providecommand{\MR}{\relax\ifhmode\unskip\space\fi MR }
\providecommand{\MRhref}[2]{%
  \href{http://www.ams.org/mathscinet-getitem?mr=#1}{#2}
}
\providecommand{\href}[2]{#2}

\end{document}